\documentclass[11pt]{article}
\usepackage{amsfonts}
\usepackage{amssymb,amsmath,amsthm}
\usepackage{epsfig}
\usepackage{color}
\usepackage{latexsym}
\usepackage{enumerate}
\usepackage{algorithm}
\usepackage{algorithmic}

\usepackage{tweaklist}

\renewcommand{\enumhook}{\setlength{\topsep}{2pt}%
  \setlength{\itemsep}{-2pt}}
\renewcommand{\itemhook}{\setlength{\topsep}{2pt}%
  \setlength{\itemsep}{-2pt}}
\usepackage[dvips, paper=letterpaper, top=1in, bottom=.75in, left=1in, right=1in, nohead, includefoot, footskip=.25in]{geometry}

\newtheorem{theorem}{Theorem}
\newtheorem{corollary}{Corollary}
\newtheorem{lemma}{Lemma}

\newtheorem{claim}{Claim}

\newtheorem{definition}{Definition}

\newtheorem{remark}{Remark}

\newcommand{\junk}[1]{}
\junk{

}

\newcommand{\costasnote}[1]{#1}

\newcommand{\notshow}[1]{{}}

\definecolor{MyGray}{rgb}{0.8,0.8,0.8}

\makeatletter
\newtheorem*{rep@theorem}{\rep@title}
\newcommand{\newreptheorem}[2]{%
\newenvironment{rep#1}[1]{%
 \def\rep@title{#2 \ref{##1}}%
 \begin{rep@theorem}}%
 {\end{rep@theorem}}}
\makeatother

\newreptheorem{theorem}{Theorem}

\begin{document}

\title{The Complexity of Optimal Mechanism Design}

\author {
Constantinos Daskalakis\thanks{Supported by a Sloan Foundation Fellowship, a Microsoft Research Faculty Fellowship, and NSF Award CCF-0953960 (CAREER) and CCF-1101491.}\\
EECS, MIT \\
\tt{costis@mit.edu}
\and
Alan Deckelbaum\thanks{Supported by Fannie and John Hertz Foundation Daniel Stroock Fellowship and NSF Award CCF-1101491.}\\
Math, MIT\\
\tt{deckel@mit.edu}
\and
Christos Tzamos\thanks{Supported by NSF Award CCF- 1101491.}\\
EECS, MIT\\
\tt{tzamos@mit.edu}
}
\addtocounter{page}{-1}
\maketitle

\begin{abstract}
Myerson's seminal work provides a computationally efficient revenue-optimal auction for selling one item to multiple bidders~\cite{Myerson81}. Generalizing this work to selling multiple items at once has been a central question in economics and algorithmic game theory, but its complexity has remained poorly understood. We answer this question by showing that a revenue-optimal auction in multi-item settings cannot be found and implemented computationally efficiently, unless ${\tt ZPP} \supseteq P^{\# {\tt P}}$. This is true even for  a single additive bidder whose values for the items are independently distributed on two rational numbers with rational probabilities. Our result is very general: we show that it is hard to compute {\em any encoding} of an optimal auction {\em of any format} (direct or indirect, truthful or non-truthful) that can be implemented in expected polynomial time. In particular, under well-believed complexity-theoretic assumptions, revenue-optimization in very simple multi-item settings can only be tractably approximated. 

We note that our hardness result applies to {\em randomized mechanisms} in a very simple setting, and is not an artifact of introducing combinatorial structure to the problem by allowing correlation among item values, introducing combinatorial valuations, or  requiring the mechanism to be deterministic (whose structure is readily combinatorial). Our proof is enabled by a flow-interpretation of the solutions of an exponential-size linear program for revenue maximization with an additional supermodularity constraint.


\end{abstract}
\thispagestyle{empty}

\newpage

\section{Introduction} \label{sec:intro}

Consider the problem facing a seller who is looking to sell items to bidders {given distributional information about their valuations, i.e.~how much they value each subset of the items}. This problem, called the {\em optimal mechanism design problem} (explained formally in Section~\ref{sec:interpretation}), has been a central problem in economics for a few decades and has recently received considerable attention in computer science as well. The special case of a single item is well-understood: in his seminal paper~\cite{Myerson81}, Myerson showed that the optimal single-item auction takes a simple form---when the values of the bidders for the item are independent. Moreover, if every bidder's value distribution is supported on rational numbers with rational probabilities, the optimal auction can be computed and implemented computationally efficiently---see e.g.~\cite{CaiDW12}. On the other hand, the general (multi-item)  problem has been open since Myerson's work and has remained poorly understood despite significant research effort devoted to its solution---see~\cite{ManelliV07} and its references for work on this problem by economists.


More recently, the problem received considerable attention in computer science where algorithmic ideas have provided insights into its structure and computational complexity. Results have come in two flavors: approximations~\cite{ChawlaHK07,ChawlaHMS10,BhattacharyaGGM10,Alaei11,KleinbergW12,DobzinskiFK11,HartN12,CaiDW13}, guaranteeing a constant fraction of the optimal revenue, and exact solutions~\cite{DW12,CaiDW12,AlaeiFHHM12,CaiDW12b}, guaranteeing full revenue extraction. These results provide  comprehensive solutions to the problem, even under very general allocation constraints~\cite{CaiDW12b,CaiDW13}, but do not pin down its computational complexity. In particular, all known exact mechanisms can be computed and implemented in time polynomial in the total number of valuations a bidder may have---i.e. the size of the support of each bidder's valuation-distribution; so they are computationally efficient only when the valuation distributions are provided explicitly in the specification of the problem, by listing their support together with the probability assigned to each valuation in the support.
 But this is not always the computationally meaningful way to describe the problem. The most trivial setting where this issue arises is that of a single additive bidder\footnote{An {\em additive bidder} is described by a value-vector $(v_1,\ldots,v_n)$, where $v_i$ is her value for item $i \in [n]$, so that her value for a subset $S$ of items is $\sum_{i \in S} v_i$.} whose values for the items are independent of support two. In this case, the bidder may have one of $2^n$ possible valuations, where $n$ is the number of items, and explicitly describing her valuation-distribution would require $\Omega(2^n)$ vectors and probabilities. However, a mechanism with complexity polynomial in $\Omega(2^n)$ is clearly inefficient, and one would like to pay complexity polynomial in the distribution's natural description complexity, i.e. the bits required to specify the distribution's $n$ marginals over the items. Such efficient mechanisms have not been discovered, except in item-symmetric settings~\cite{DW12}.

To sum up, previous work has provided solutions to the optimal mechanism design problem in broad multi-item settings~\cite{CaiDW12,AlaeiFHHM12,CaiDW12b}, but these solutions fall short of characterizing the computational complexity of the problem, and so do existing lower bound approaches~\cite{Briest08,PapadimitriouP11,DobzinskiFK11,DaskalakisDT12} (see broader discussion for the latter in Section~\ref{sec:related}). In this paper, we answer this question by showing that optimal mechanism design is computationally intractable, even in the most basic settings.


\begin{theorem} \label{thm:main}
There is no expected polynomial-time solution to the optimal mechanism design problem (formal definition in Section~\ref{sec:interpretation}) unless ${\tt ZPP} \supseteq {\tt P}^{\# {\tt P}}$.

Moreover, it is $\# {\tt P}$-hard to determine whether every optimal mechanism assigns a specific item to a specific type of bidder with probability $0$ or with probability $1$ \costasnote{(at the Bayesian Nash equilibrium of the mechanism),} given the promise that one of these two cases holds simultaneously for all optimal mechanisms.

The above are true even in the case of selling multiple items to a single additive, quasi-linear\footnote{A bidder is {\em quasi-linear} if her utility for purchasing a subset $S$ of the items at price $p_S$  is $v_S - p_S$, where $v_S$ is her value for $S$.} bidder, whose values for the items are independently distributed on two rational numbers with~rational~probabilities. \end{theorem}

\costasnote{\noindent 
The contribution of our result is two-fold. First, it gives a definitive proof that approximation is {\em necessary} for revenue optimization beyond Myerson's single-item setting. Approximation has been heavily used in algorithmic work on the problem, but there has been no justification for its use, at least in simple settings that don't induce combinatorial structure in the valuations of the bidders (sub-modular, OXS, etc.), or the allocation constraints of the setting. Second, our result represents advancement in our techniques for showing lower bounds for optimal mechanism design. Despite evidence that the structure of the optimal mechanism is complicated even in simple settings (see Section~\ref{sec:techniques}), previous work has not been able to harness this evidence to obtain computational hardness results. Our approach, using duality theory and flows to narrow into a family of instances for which this is possible, provides a new paradigm for proving hardness results in optimal mechanism design and, we expect, outside of algorithmic game theory. Again we note that complexity creeps in not because we assume correlations in the item values (which can easily introduce combinatorial dependencies among them), or because we restrict attention to deterministic mechanisms (whose structure is readily combinatorial), but because our approach reveals combinatorial structure in the optimal mechanism, which can be exploited for a reduction. See further discussion in Sections~\ref{sec:techniques} and~\ref{sec:related}.}

\subsection{The Optimal Mechanism Design Problem}  \label{sec:interpretation}

\costasnote{To explain our result, we describe the {\em Optimal Mechanism Design (OMD) problem} more formally. As we are aiming for a broad lower bound we take the most general approach on the definition of the problem placing no constraints on the form of the sought after mechanism, or how this mechanism is meant to be described. Hence, our lower bounds apply to computing direct revelation mechanisms as well as any conceivable type of mechanism.}

\medskip \noindent  {\sc Input:} This consists of the names of the items and the bidders, the allocation constraints of the setting (specifying, e.g., that an item may be allocated to at most one bidder, etc.), and a probability distribution on the valuations, or {\em types}, of the bidders. The { type of a bidder} incorporates information about how much she values every subset of the items, as well as what utility she derives for receiving a subset at a particular price. For example, the type of an additive quasi-linear bidder can be encapsulated in a vector of values (one value per item). We won't make any assumptions about how the allocation constraints are specified. In general, these could either be hard-wired to a family of instances of the OMD problem, or provided as part of the input in a computationally meaningful way. For the purposes of our intractability results, the allocation constraints will be trivial, enforcing that we can only allocate at most one copy of each item, and we restrict our attention to  instances with precisely these allocation constraints. As far as the type distribution is concerned, we restrict our attention to  additive quasi-linear bidders with independent values for the items. So, for our purposes, the type distribution of a bidder is specified by specifying its marginal on each item. We assume that each marginal is given explicitly, as a list of the possible values for the item as well as the probabilities assigned to each value.\footnote{There are of course other ways to describe these marginals. For example, we may only have sample access to them, or we may be given a circuit that takes as input a value and outputs the probability assigned to that value. As our goal is to prove lower bounds, the assumption that the marginals are provided explicitly in the input only makes the lower bounds stronger.}


\smallskip \noindent {\sc Desired Output:} The goal is to compute a (possibly randomized) auction that optimizes, over all possible auctions, the expected revenue of the auctioneer, i.e. the expected sum of  prices paid by the bidders at the Bayes Nash equilibrium of the auction,\footnote{Informally {\em Bayesian Nash equilibrium} is the extension of {\em Nash equilibrium} to incomplete-information games, i.e. games where the utilities of players are sampled from a probability distribution. We won't provide a formal definition as it is quite involved and is actually not required for our lower bounds, which focus on the single-bidder case.

For the purposes of the problem definition though, we note that, if an auction has multiple Bayesian Nash equilibria, its revenue is not well-defined as it may depend on what Bayesian Nash equilibrium the bidders end up playing. So we would like to avoid such auctions given the uncertainty about their revenue. Again this complication won't be relevant for our results as all auctions we construct in our hardness proofs will have a unique Bayes Nash equilibrium.} where the expectation is taken with respect to the bidders' types, the randomness in their strategies (if any) at the Bayes Nash equilibrium, as well as any internal randomness that the auction uses. 

We note that there is a large universe of possible auctions with widely varying formats, e.g.~sealed envelope auctions, dynamic auctions, all-pay auctions, etc. And there could be different auctions with optimal expected revenue. As our goal is to prove robust intractability results for OMD, we take a general approach imposing no restrictions on the format of the auction, and no restrictions on the way the auction is encoded. The encoding should specify in a computationally meaningful way what actions are available to the bidders, how the items are allocated depending on the actions taken by the bidders, and what prices are charged to them, where both allocation and prices  could be outputs of a randomized function of the bidders' actions. In particular, a computationally efficient solution to OMD induces the following:
\begin{itemize}
\item[] \costasnote{ \smallskip \noindent {\sc Auction Computation\&Simulation:} A { computationally efficient solution} to a family ${\cal I}$ of OMD problems induces a pair of algorithms $\cal C$ and $\cal S$ satisfying the following:
\begin{enumerate}
\item {[{\em auction computation}]} ${\cal C}: {\cal I} \rightarrow {\cal E}$ is an expected polynomial-time algorithm mapping instances $I \in {\cal I}$ of the OMD problem to auction encodings ${\cal C}(I) \in {\cal E}$; e.g. ${\cal C}(I)$ may be ``second price auction'', or ``English auction with reserve price \$5'', etc.\label{tractability 1}
\item {[{\em auction simulation}]} ${\cal S}$ is an expected polynomial-time algorithm mapping instances $I \in {\cal I}$ of the OMD problem, encodings ${\cal C}(I)$ of the optimal auction for $I$, and realized types $t_1,\ldots,t_m$ for the bidders, to a sample from the (possibly randomized) allocation and price rule of the auction  encoded by ${\cal C}( I)$, at the Bayes Nash equilibrium of the auction when the types of the bidders are $t_1,\ldots,t_m$.\label{tractability 2}
\end{enumerate}}
\end{itemize}
\noindent Clearly, Property~\ref{tractability 1} holds because  computing the optimal auction encoding ${\cal C}(I)$ for an instance ${I}$ is assumed to be efficient. But why does there exist a simulator ${\cal S}$ as in Property~\ref{tractability 2}? 
Well, when the auction ${\cal C}({ I})$ is executed, then somebody (either the auctioneer, or the bidders, or both) need to do some computation: the bidders need to decide how to play the auction (i.e. what actions from among the available ones to use), and then the auctioneer needs to allocate the items and charge the bidders. In the special case of a direct Bayesian Incentive Compatible mechanism,\footnote{A mechanism is {\em direct} if the only available actions to a bidder is to report a type in the support of her type-distribution. A direct mechanism is {\em Bayesian Incentive Compatible} if it is a Bayes Nash equilibrium for every bidder to truthfully report her type to the auctioneer.} the bidders need not do any computation, as it is a Bayes Nash equilibrium strategy for each of them to truthfully report their type to the auctioneer. In this case, all the computation must be done by the auctioneer who needs to sample from the (possibly randomized) allocation and price rule of the mechanism given the bidders' reported types. In general (possibly non-direct, or multi-stage) mechanisms, both bidders and auctioneer may need to do some computation: the bidders need to compute their Bayes Nash equilibrium strategies given their types, and the auctioneer needs to sample from the (possibly random) allocation and price rule of the mechanism given the bidders' strategies. These computations must all be computationally efficient, as otherwise \costasnote{the execution of the auction ${\cal C}(I)$ would not be computationally efficient.} Hence an efficient solution must induce an efficient simulator ${\cal S}$. \costasnote{Notice, in particular, that ${\cal S}$ requires that the combined computation performed by bidders and auctioneer be computationally efficient. This is important as placing no computational restrictions on the bidder side easily leads to spurious ``efficient'' solutions to OMD, as discussed in Section~\ref{sec:powerfull bidders}.}

\smallskip In view of the above discussion, Theorem~\ref{thm:main}
 establishes that,  even in very simple special cases of the OMD problem, there does not exist  a pair $({\cal C}, {\cal S})$ of efficient auction computation and simulation algorithms, i.e. the optimal auction cannot be found computationally efficiently, or cannot be executed efficiently, or both. We note that our hardness result is subject to the assumption ${\tt ZPP} \not\supseteq {\tt P}^{\# {\tt P}}$ (rather than ${\tt P} \not\supseteq {\tt P}^{\# {\tt P}}$) solely because we prove lower bounds for randomized mechanisms. 
 
\costasnote{ 
 \begin{remark}[Hardness of BIC Mechanisms] {A lot of research on optimal mechanism design has focused on finding optimal Bayesian Incentive Compatible (BIC) mechanisms,  as focusing on such mechanisms costs nothing in revenue due to the direct revelation principle (see~\cite{AGTbook} and Section~\ref{sec:prelim}). As an immediate corollary of Theorem~\ref{thm:main} we obtain that it is {\tt \#P}-hard to compute the (possibly randomized) allocation and price rule of the optimal BIC mechanism (Corollary~\ref{cor:BIC hardness} in Section~\ref{sec:prelim}). However, Theorem~\ref{thm:main} is much broader, in two respects: a. in the definition of the OMD problem we impose no constraints on what type of auction should be found; b. we don't require an explicit computation of the (possibly randomized) allocation and price rule of the mechanism, but allow an expected polynomial-time algorithm that samples from the allocation~and~price~rule.}
  \end{remark}}

%
\subsection{Techniques} \label{sec:techniques}


There are serious obstacles in establishing intractability results for optimal mechanism design, the main one being that the structure of optimal mechanisms is poorly understood even in very simple settings. To prove Theorem~\ref{thm:main}, we need to find a family of mechanism design instances whose optimal solutions are sufficiently complex to enable reductions from some {\tt \#P}-hard problem, while at the same time are sufficiently restricted so that solutions to the {\tt \#P}-hard problem can actually be extracted from the structure of the optimal mechanism. However, there is no apparent structure in the optimal mechanism even in the simple case of a single additive bidder. 

To gain some intuition, it is worth pointing out that the optimal mechanism for selling multiple items to an additive bidder is not necessarily comprised of the optimal mechanisms for selling each item separately. Here is an example from Hart and Nisan~\cite{HartN12}. Suppose that there are two items and the bidder's value for each is either $1$ or $2$ with probability ${1 \over 2}$, independently of the other item. In this case, the maximum expected revenue from selling the items separately is $2$, achieved~e.g.~by posting a price of $1$ on each item. However, offering instead the bundle of both items at a price of $3$ achieves a revenue of $3 \cdot {3 \over 4} = {9 \over 4}.$

On the other hand, bundling the items together is not always better than selling them separately. If there are two items with values $0$ or $1$ with probability $1 \over 2$, independently from each other, then selling the bundle of the two items achieves revenue at most $3 \over 4$, but selling the items separately yields revenue of~$1$, which is optimal.

In general, the optimal mechanism may have much more intricate structure than either selling the grand bundle of all the items, or selling each item separately, even when the item values are i.i.d. In fact, the optimal mechanism might not even be deterministic: we may need to offer for sale not only sets of items, but also lotteries over sets of items. Here is an example. Suppose that there are two items, one of which is distributed uniformly in $\{1,2\}$ and the other uniformly in $\{1,3\}$, and suppose that the items are independent. In this case, the optimal mechanism offers the bundle of the two items at price $4$, and it also offers at price $2.5$ a lottery that with probability $1 \over 2$ gives both the items and with probability $1 \over 2$ just the first item.\footnote{That this is the optimal mechanism  follows from our techniques; see Section~\ref{sec:end lp theorem}.}

\smallskip Given that optimal mechanisms have no apparent structure, even in the simple case of an additive bidder with independent values for the items, the main challenge in proving our hardness result is pinning down a family of sufficiently interesting instances of the problem for which we can still characterize the form of the optimal mechanism. To do this we follow a principled approach starting with a folklore, albeit exponentially large, linear program for revenue optimization, constructing a relaxation of this LP, and showing that, in a suitable class of instances, the solution of the relaxed LP is also a solution to the original LP, which has rich enough structure to embed a {\tt \#P}-hard problem. In more detail, our approach is the following:

\begin{itemize}
\item In Section~\ref{LP1} we present LP1, the folklore, albeit exponentially large, linear program for computing a revenue optimal auction.
\item In Section~\ref{LP2} we relax the constraints of LP1 to construct a new, still exponentially large, linear program LP2. The solutions of the relaxed LP need not provide solutions to the original mechanism design problem. We prove however that an optimal LP2 solution is indeed an optimal LP1 solution if it happens to be monotone and supermodular.
\item In Section~\ref{LP3} we take LP3, the dual program to LP2. We interpret its solutions as solutions to a minimum-cost flow problem on a lattice.
\item In Section~\ref{canonicalsolution} we characterize a \emph{canonical solution} to a specific subclass of LP3 instances. This solution requires ordering of the subset sums of an appropriate set of integers.
\item In Section~\ref{solvingrelaxed} we use duality to convert a canonical LP3 solution to a unique LP2 solution. We are therefore able to characterize the unique solution for a variety of LP2 instances. 
\item In Section~\ref{lp2tolp1} we show that the LP2 solutions obtained above are also feasible and optimal for the corresponding LP1 instance. Thus, we gain the ability to characterize unique optimal solutions of a class of LP1 instances.
\item In Section~\ref{hardnessproof} we show how to encode a {\tt \#P}-hard problem into the class of LP1 instances that we have developed.
\end{itemize}

\subsection{Related Work}\label{sec:related} We have already provided background on optimal mechanism design in economics literature, as well as algorithmic solutions to the problem. We summarize the state of affairs in that exact, computationally efficient solutions are known for broad multi-item settings, as long as the bidder types are independent and their type-distributions are given explicitly. In particular, the runtime of such solutions is polynomial in the size of the support of each bidder's type-distribution, which may be exponential in the natural description complexity of the distribution, e.g., when the distribution is product over the items.

There has also been considerable effort towards computational lower bounds for the problem. Nevertheless, all known  results are for either somewhat exotic families of valuation functions or distributions over valuations, or for computing the optimal deterministic, as opposed to general (possibly randomized) auction. In the first vein, Dobzinski et al.~\cite{DobzinskiFK11} show that optimal mechanism design for OXS bidders is ${\tt NP}$-hard via a reduction from the {\sc Clique} problem. OXS valuations are described implicitly via a graph, include additive, unit-demand\footnote{A {\em unit-demand} bidder is described by a vector $(v_1,\ldots,v_n)$ of values, where $n$ is the number of items. If the bidder receives item $i$, her value is $v_i$, while if she receives a set of more than one item, her value for that set is her value for her favorite item in the set.} and much broader classes of valuations, and are more amenable to lower bounds given the combinatorial nature of their definition. (See further discussion in Section~\ref{beyondadditive}.) In very recent work~\cite{DaskalakisDT12} the authors of the present paper prove {\tt \#P}-hardness of the optimal mechanism design problem for a single item and a single bidder whose value for the item is the sum of independently distributed attributes.
Compared to these lower bounds, our goal in this paper is instead to prove intractability results for very simple valuations (namely additive) and distributions over them (namely the item values are independent and the distribution of each item is given explicitly), and which have no combinatorial structure incorporated in their definition.

On the complexity of optimal deterministic auctions, Briest \cite{Briest08} shows  inapproximability results for selling multiple items to a single unit-demand bidder via an item-pricing auction,\footnote{An {\em item-pricing auction} posts a price for each item, and lets the bidder buy whatever item s/he wants. It is clearly a deterministic auction as there is no randomness used by the auctioneer. A randomized auction for this setting could also be pricing lotteries over items.} when the bidder's values for the items are correlated according to some explicitly given distribution. More recently, Papadimitriou and Pierrakos~\cite{PapadimitriouP11} show ${\tt APX}$-hardness results for the optimal, incentive compatible, deterministic auction when a single item is sold to multiple bidders, whose values for the item are correlated according to some explicitly given distribution. (We note that the settings considered in both Briest~\cite{Briest08} and Papadimitriou-Pierrakos~\cite{PapadimitriouP11} are polynomial-time solvable via linear programming when the determinism requirement is removed~\cite{DobzinskiFK11,DW12}.) Finally, in~\cite{DaskalakisDT12} the authors of the present paper provide {\sc SqrtSum}-hardness results for the optimal item-pricing problem when there is a single unit-demand bidder whose values for the items are independent of support two, and when either the values or the probabilities may be irrational. Compared to these results, our lower bounds apply to general (i.e. possibly randomized) auctions, the values of the items are distributed independently, and both supports and probabilities are rational numbers.

\subsection{Further Results and Discussion} In addition to our main theorem, Theorem~\ref{thm:main}, we show that, beyond additive valuations, computational intractability arises for very simple settings with submodular valuations, namely even for a single, budget-additive,\footnote{A {\em budget-additive} bidder is described by a value-vector $(v_1,\ldots,v_n)$, where $v_i$ is her value for item $i \in [n]$, and a budget $B$ so that her value for a subset $S$ of items is $\min\{\sum_{i \in S} v_i,B\}$.} quasi-linear bidder whose values for the items are perfectly known, but there is uncertainty about her budget. This is presented as Theorem~\ref{thm:main2} in Section~\ref{sec:budget additive bidders}. 

Interestingly, the hard instances constructed in the proof of Theorem~\ref{thm:main2} have trivial indirect mechanism descriptions, but require {\tt NP} power for the bidder to determine what to buy. We show that this phenomenon is quite general, discussing how easily implementable, optimal, indirect mechanisms may be trivial to compute, if the bidders are assumed sufficiently  powerful computationally. This is discussed in Section~\ref{sec:powerfull bidders}, and further justifies our framework (in Section~\ref{sec:interpretation}) for studying the complexity of optimal mechanism design.

\section{Preliminaries}\label{sec:prelim}

Throughout this paper, we restrict our attention to instances of the mechanism design problem where a seller wishes to sell a set $N = \{1,2, \ldots, n\}$ of items to a single additive quasi-linear bidder whose values for the items are independent of support $2$. In particular, the bidder values item $i$ at $a_i$, with probability $1-p_i$, and at $a_i + d_i$, with probability $p_i$, independently of the other items, where $a_i$, $d_i$, and $p_i$ are positive rational numbers. If she values $i$ at $a_i$, we say that her value for $i$ is ``low'' and, if she values it at $a_i+d_i$, we say that her value is ``high.'' The specification of the instance comprises the numbers $\{a_i,d_i,p_i\}_{i=1}^n$. 

A mechanism $\cal M$ for an instance ${ I}$ as above specifies a set ${\cal A}$ of actions available to the bidder together with a rule mapping each action $\alpha \in {\cal A}$ to a (possibly randomized) allocation $A_{\alpha} \in \{0,1\}^n$, specifying which items the bidder gets, and a (possibly randomized) price \costasnote{$\tau_{\alpha} \in \mathbb{R}$} that the bidder pays, where $A_\alpha$ and $\tau_\alpha$ could be correlated. Facing this mechanism, a bidder whose values for the items are instantiated to some vector $\vec{v} \in \times_i \{a_i, a_i+d_i\}$ chooses any action in $\arg \max_{\alpha \in {\cal A}}\{ \vec{v} \cdot \mathbb{E}(A_\alpha) - \mathbb{E}(\tau_\alpha)\}$ or any distribution on these actions, as long as the maximum is non-negative, since $\vec{v} \cdot \mathbb{E}(A_\alpha) - \mathbb{E}(\tau_\alpha)$ is the expected utility of the bidder for choosing action $\alpha$.\footnote{If the maximum utility under $\vec{v}$ is negative, the bidder would ``stay home.'' To ease notation, we can include in ${\cal A}$ a special action ``stay home'' that results in the bidder getting nothing and paying nothing. If all other actions give negative utility the bidder could just use this special action.} In particular, any such choice is a {\em Bayesian Nash equilibrium} behavior for the bidder, in the degenerate case of a single bidder that we consider.\footnote{As mentioned earlier, we opt not to define Bayesian Nash equilibrium for multi-bidder mechanisms as this definition is involved and irrelevant for our purposes.} If for all vectors $\vec{v}$ there is a unique optimal action $\alpha_{\vec{v}} \in {\cal A}$ in the above optimization problem, then mechanism ${\cal M}$ induces a mapping from valuations $\vec{v} \in \times_i \{a_i, a_i+d_i\}$ to (possibly randomized) allocation and price  pairs $(A_{\alpha_{\vec{v}}}, \tau_{\alpha_{\vec{v}}})$. If there are $\vec{v}$'s with non-unique maximizers, then we break ties in favor of the action with the highest $\mathbb{E}(\tau_\alpha)$ and, if there are still ties, lexicographically beyond that.\footnote{{We can enforce this tie-breaking with an arbitrarily small hit on revenue as follows: For all $\alpha$, we decrease the (possibly random price) $\tau_\alpha$ output by ${\cal M}$ by a well-chosen amount---think of it as a rebate---which gets larger as $\mathbb{E}(\tau_\alpha)$ gets larger. We can choose these rebates to be sufficiently small so that they only serve the purpose of tie-breaking. These rebates won't affect our lower bounds.}}

In Section~\ref{sec:interpretation} we explained in detail what it means to solve an instance ${ I}=\{a_i,d_i,p_i\}_{i=1}^n$ computationally efficiently. In short, the solution needs to encode the action set ${\cal A}$ as well as the mapping from actions $\alpha \in {\cal A}$ to $(A_\alpha,\tau_\alpha)$, in a way that given an instantiated type $\vec{v} \in \times_i \{a_i, a_i+d_i\}$ we can computationally efficiently sample from $A_{\alpha_{\vec{v}}}$ and from $\tau_{\alpha_{\vec{v}}}$.

\smallskip It is convenient for our arguments to first study {\em direct} mechanisms, where the action set available to the bidder coincides with her type space $\times_i \{a_i, a_i+d_i\}$. In this case, we can equivalently think of the actions available to the bidder as declaring any subset $S \subseteq N$, where the correspondence between subsets and value vectors is given by $\vec v(S) = \sum_{i \in N} a_i e_i + \sum_{i \in S} d_i e_i$ where $e_i$ is the unit vector in dimension $i$; i.e.~$S$ represents the items whose values are high. For all actions $S \subseteq N$, the mechanism induces a vector $\vec q(S) \in [0,1]^n$ of probabilities that the bidder receives each item and an expected price $\tau(S) \in \mathbb{R}$ that the bidder pays. The expected utility of a bidder of type $\vec{v}(S)$ for choosing action $S'$ is given by $\vec v(S) \cdot \vec q(S') - \tau(S')$. We denote by $u(S)=\vec v(S) \cdot \vec q(S) - \tau(S)$ her expected utility for reporting her true type. The following are important properties of direct mechanisms.

\begin{definition}
A direct mechanism for the family of single-bidder instances we consider in this paper is Bayesian Incentive Compatible (BIC) if the bidder cannot
benefit by misreporting the set of items he values highly. Formally:
\[\forall S,T\subseteq N: \vec v(S) \cdot \vec q(S) - \tau(S) \ge \vec v(S) \cdot \vec q(T) - \tau(T).\]
Or, equivalently:
\[\forall S,T\subseteq N: u(S) \ge u(T) + (\vec v(S) - \vec v(T)) \cdot \vec q(T).\]
\end{definition}

\begin{definition}
The mechanism is individually rational (IR) if $u(S) \ge 0$, for all $S \subseteq N$.
\end{definition}

\costasnote{We note that an immediate consequence of our proof of Theorem~\ref{thm:main} is the following.
 \begin{corollary} \label{cor:BIC hardness}
Given an instance $I=\{a_i,d_i,p_i\}_{i=1}^n$ of the optimal mechanism design problem (i.e. the same setting as Theorem~\ref{thm:main}), it is {\tt \#P}-hard to compute circuits (or any other implicit but efficient to evaluate description of) $q: \times_i \{a_i, a_i+d_i\} \rightarrow [0,1]^n$ and $\tau: \times_i \{a_i, a_i+d_i\} \rightarrow \mathbb{R}$ inducing a revenue-optimal, BIC, IR mechanism.
 \end{corollary}
\noindent Note that, for a meaningful lower bound, we cannot ask for $q(S), \tau(S)$ for all $S$, as there are too many $S$'s---namely $2^n$. Instead we need to ask for some implicit but computationally meaningful description of them, such as in the form of circuits, which can be evaluated on an input $S$ in time polynomial in their size and the number of bits required to describe $S$---if we don't require that $q$ and $\tau$ can be evaluated efficiently on any given $S$ we would allow for trivial solutions such as ``$I$ is itself an implicit description of the optimal $q$ and $\tau$ for $I$.'' We conclude with the following remark.}

\begin{remark}
For the single-bidder instances we consider in this paper, a direct mechanism that is Bayesian Incentive Compatible is also {\em Incentive Compatible} and vice versa.\footnote{{For the reader who is not familiar with these concepts here is a brief explanation. A direct multi-bidder mechanism is {Incentive Compatible} if, for all bidders $i$ and all $t_i$, it is optimal for bidder $i$ to truthfully report $t_i$ no matter what the other bidders report. It is a stronger notion than {Bayesian Incentive Compatibility} which requires that, for all bidders $i$ and all $t_i$, it is optimal for bidder $i$ to truthfully report $t_i$, if the other bidders also truthfully report their types in expectation with respect to their realized types. I.e.~the latter concept requires that it is Bayesian Nash equilibrium for every bidder to truthfully report their type, while the former concept requires that it is a Dominant Strategy equilibrium for every bidder to truthfully report their type. Clearly, the two concepts collapse if there is just one bidder.}} As all our hardness results are for single-bidder instances, they simultaneously show the intractability of computing optimal Bayesian Incentive Compatible as well as optimal Incentive Compatible mechanisms.
\end{remark}
\costasnote{\subsection{Our Proof Plan} To prove Theorem~\ref{thm:main} we narrow into a family of single-bidder instances (of the form defined in the beginning of this section) for which there is a unique optimal BIC, IR, direct mechanism, and this mechanism satisfies the following: for a special item $i^*$ and a special type $S^*$, $q_{i^*}(S^*) \in \{0,1\}$ but it is {\tt \#P}-hard to decide whether $q_{i^*}(S^*) = 0$. Our approach for narrowing into this family of instances was outlined in Section~\ref{sec:techniques} and is provided in detail in Sections~\ref{sec:LP} through~\ref{hardnessproof}. These sections establish two hardness results: First, they show Corollary~\ref{cor:BIC hardness}  that it is {\tt \#P}-hard to compute circuits for $q$ and $\tau$, as if we could compute such circuits in polynomial time then we would also be able to answer whether $q_{i^*}(S^*) = 0$ in polynomial time. Second, they establish Theorem~\ref{thm:main} restricted to BIC mechanisms. Indeed, if there were expected polynomial-time algorithms ${\cal C}$, ${\cal S}$ (see Section~\ref{sec:interpretation}), then we would plug in type $S^*$ and get one sample from the allocation rule for that type. Notice that this sample will allocate item $i^*$ if and only if $q_{i^*}(S^*) = 1$, which is {\tt \#P}-hard to decide. Details are provided in Section~\ref{hardnessproof}.

Finally, it is straightforward to translate this hardness result to general mechanisms by making the following observation, called the ``Revelation Principle''~\cite{AGTbook}:
\begin{itemize}
\item Any (possibly non-direct) mechanism ${\cal M}$ has an equivalent BIC, IR, direct mechanism ${\cal M'}$ so that the two mechanisms induce the exact same mapping from types $\vec{v}$ to (possibly randomized) allocation and price pairs. Indeed, as explained above, given the (possibly randomized) allocation rule $A$ and price rule $\tau$ of ${\cal M}$ we can define the mapping $\vec{v} \mapsto (A_{\alpha_{\vec{v}}}, \tau_{\alpha_{\vec{v}}})$. This mapping is in fact itself a BIC, IR direct mechanism ${\cal M'}$. Clearly, ${\cal M}$ and ${\cal M}'$ have the same expected revenue.

\end{itemize}
As mentioned above, the hard instances we narrow into in our proof satisfy that their optimal BIC, IR direct mechanism is unique and it is a {\tt \#P}-hard problem to  tell whether $q_{i^*}(S^*) = 0$ or $1$ for a special item~$i^*$ and a special type $S^*$. The above observation implies that, for any instance in our family, any (possibly non-direct) optimal mechanism ${\cal M}$ for this instance needs to induce an optimal direct mechanism. Since the latter is unique, ${\cal M}$ needs to give the same (possibly randomized) allocation to type $S^*$ that the unique direct mechanism does. As argued above, getting one sample from this allocation allows us to decide whether $q_{i^*}(S^*) = 1$, a {\tt \#P}-hard problem. As any efficient solution to our family of OMD instances would allow us to get samples from the allocation rule of an optimal mechanism in expected polynomial-time, we would be able to solve a {\tt \#P}-hard problem in expected polynomial-time. This establishes Theorem~\ref{thm:main} (for unrestricted mechanisms).}



\section{A Linear Programming Approach} \label{sec:LP}
Our goal in Sections~\ref{sec:LP} through~\ref{hardnessproof} is showing that it is computationally hard to compute a BIC, IR direct mechanism that maximizes the seller's expected revenue, even in the single-bidder setting introduced in Section~\ref{sec:prelim}. In this section, we define three exponential-size linear programs which are useful for zooming into a family of hard instances that are also amenable to analysis. Our LPs are defined using the notation introduced in Section~\ref{sec:prelim}.
\subsection{Mechanism Design as a Linear Program}\label{LP1}
The optimal BIC and IR mechanism for the family of single-bidder instances introduced in Section~\ref{sec:prelim} can be found by solving the following linear program, which we call LP1:

\begin{figure}[h!]
\begin{center}
\fbox{\begin{minipage}{13cm}
\[
\textrm{max }  E_S[ \vec v(S) \cdot \vec q(S) - u(S) ]
\]
\[
	\begin{array}{lcr}
	\textrm{subject to :} & \\
	\textrm{$\forall S,T\subseteq N$ : } & u(S) \ge u(T) + (\vec v(S) - \vec v(T)) \cdot \vec q(T) & \textrm{ (BIC)}\\
	\textrm{$\forall S\subseteq N$ : } & u(S) \ge 0  & \textrm{ (IR)}\\
	\textrm{$\forall S\subseteq N, i \in N$ : } & 0 \le q_i(S) \le 1 & \textrm{ (PROB)}
	\end{array}
\]
\end{minipage}}
\caption{LP1, the linear program for revenue maximization}
\end{center}
\end{figure}
\vspace{-10pt}\noindent Notice that the expression $\vec v(S) \cdot \vec q(S) - u(S)$ in the objective function equals the price $\tau(S)$ that the bidder of type $S$ pays when reporting $S$ to the mechanism. The expectation is taken over all $S \subseteq N$, where the probability of set $S$ is given by $p(S) = \prod_{i \in S}p_i \cdot \prod_{j \notin S}(1-p_j)$. We notice that this program has exponential size (variables and constraints).

\subsection{A Relaxed Linear Program}\label{LP2}

We now remove constraints from LP1 and perform further simplifications, making the program easier to analyze. Later on we identify a subclass of instances where optimal solutions to the relaxed program induce optimal solutions to the original program {(see Lemma~\ref{lem:supermodulariy})}.

As a first step, we relax LP1 by considering only BIC constraints that correspond to neighboring types (types that differ in one element). We also drop the constraint that the probabilities $q_i(S)$ are non-negative:

\[
\textrm{max }  E_S[ \vec v(S) \cdot \vec q(S) - u(S) ]
\]
\[
	\begin{array}{lcr}
	\textrm{subject to :} & \\
	\textrm{$\forall S\subseteq N, i \notin S$ : } & u(S \cup \{i\}) \ge u(S) + d_i q_i(S) & \textrm{ (BIC1)}\\
	\textrm{$\forall S\subseteq N, i \notin S$ : } & u(S) \ge u(S \cup \{i\}) - d_i q_i(S \cup \{i\}) & \textrm{ (BIC2)}\\
	\textrm{$\forall S\subseteq N$ : } & u(S) \ge 0  & \textrm{ (IR)}\\
	\textrm{$\forall S\subseteq N, i \in N$ : } &  q_i(S) \le 1 & \textrm{(PROB')}
	\end{array}
\]

Since the coefficient of every $q_i(S)$ in the objective is strictly positive, no $q_i(S)$ can be increased in any optimal solution without violating a constraint. 
We therefore conclude the following about $q_i(S)$:
\begin{itemize}
\item If $i \in S$, then $q_i(S)$ is only upper-bounded by constraint PROB', and thus $q_i(S) = 1$ in every optimal solution.
\item If $i \notin S$, then $q_i(S) = \min\{1, \frac {u(S \cup \{i\}) - u(S)} {d_i}\}$ from (BIC1) and (PROB'). Furthermore, from (BIC2) we have   $ \frac {u(S \cup \{i\}) - u(S)} {d_i} \le q_i(S \cup \{i\}) = 1$, and thus $q_i(S) = \frac {u(S \cup \{i\}) - u(S)} {d_i}$.
\end{itemize}

So the program becomes (after setting $q_i(S) = 1$ whenever $i \in S$, removing the constant terms from the objective, and tightening the constraints (BIC1) to equality):
\[
\textrm{max }  E_S\left[ \sum_{i \notin S} v_i(S) q_i(S) - u(S) \right]
\]
\[
	\begin{array}{lcr}
	\textrm{subject to :} & \\
	\textrm{$\forall S\subseteq N, i \notin S$ : } & q_i(S) = \frac {u(S \cup \{i\}) - u(S)} {d_i}  & \textrm{ (BIC1')}\\
	\textrm{$\forall S\subseteq N, i \notin S$ : } & u(S \cup \{i\}) - u(S)  \le d_i &  \textrm{ (BIC2)} \\
	\textrm{$\forall S\subseteq N$ : } & u(S) \ge 0 & \textrm{ (IR)} \\
	\textrm{$\forall S\subseteq N, i \notin S$ : } & q_i(S) \le 1 & \textrm{(PROB')}

	\end{array}
\]

\noindent Notice that the constraint (PROB') is trivially satisfied as a consequence of (BIC1') and (BIC2).

We now rewrite the objective, substituting $q_i(S)$ according to (BIC1') and noting that $v_i(S) = a_i$ for $i \notin S$:

\begin{eqnarray*}
&E_S&\left[ \sum_{i \notin S} a_i \frac {u(S\cup\{i\}) - u(S)} {d_i} - u(S) \right] \\ 
&=& E_S \left[   u(S)\left(  -1 - \sum_{i \notin S} \frac{a_i}{d_i} + \sum_{i \in S}\left( \frac{a_i}{ d_i}\cdot \frac{ p(S \setminus \{i\})}{ p(S)}  \right)  \right) \right]
\end{eqnarray*}
obtained by grouping together all coefficients of $u(S)$, adjusting by the appropriate probabilities.

We now note that $\frac{p(S \setminus \{i\})}{p(S)} = -1 +\frac{1}{p_i} $, and our objective becomes:
\[  E_S\left[ u(S)\left( -1 - \sum_{i \in N}\frac{a_i}{d_i} +  \sum_{i \in S} \frac {a_i} {p_i d_i}  \right) \right].\]


We now perform a change of notation so that the program takes a simpler form. In particular, we set 
$$B \leftarrow \kappa\left(1 + \sum_{i \in N} \frac {a_i} {d_i}\right); \qquad x_i \leftarrow \frac {\kappa a_i} {p_i d_i}$$ 
where $\kappa$ is some positive constant, and the objective becomes\\
$\frac{1}{\kappa}E_S[ (\sum_{i \in S} x_i - B) u(S) ]$. Since $1/\kappa$ is constant, we are lead to study the following program, LP2:
\begin{figure}[h!]
\begin{center}
\fbox{\begin{minipage}{13cm}

\[
\max_u E_S[ (\sum_{i \in S} x_i - B) u(S) ]
\]
\[
	\begin{array}{lcr}
	\textrm{subject to :} & \\
	\textrm{$\forall S\subseteq N, i \notin S$ : } & u(S \cup \{i\}) - u(S)  \le d_i &  \textrm{ (BIC2)} \\
	\textrm{$\forall S\subseteq N$ : } & u(S) \ge 0 & \textrm{ (IR)} \\
	\end{array}
\]
\end{minipage}}
\caption{LP2, the relaxed linear program}
\end{center}
\end{figure}

In constructing LP2, our new constants $B$ and $x$ were a function of $\vec{p}$, $\vec{a}$, $\vec{d}$, and a newly introduced constant $\kappa$. We note that, by adjusting $\kappa$, we are able to obtain a wide range of relevant $B$ and $x$ values.
\begin{lemma}\label{newvariables}
For any  $B$, $\vec{x}$, $\vec{p}$ and $\vec{d}$ such that $B > \sum_{i \in N} p_i x_i$, there exist (efficiently computable)  $\vec{a}$ and $\kappa$ such that $B = \kappa(1 + \sum_{i \in N} \frac {a_i} {d_i})$ and $x_i =  \frac {\kappa a_i} {p_i d_i}$. If  $B$, $\vec{x}$, $\vec{p}$ and $\vec{d}$ are rational, then $\vec{a}$ and $\kappa$ are rational as well.
\end{lemma}
\begin{proof}
We want that $x_i = \frac{ \kappa a_i } {p_i d_i}$ and $ B = \kappa + \sum_{i \in N} \frac{ \kappa a_i } {d_i} = \kappa + \sum_{i \in N} p_i x_i.$
Indeed, these equalities follow from setting
\[ \kappa \leftarrow B - \sum_{i \in N} p_i x_i; \qquad a_i \leftarrow \frac{ p_i d_i x_i } {\kappa}. \]
\end{proof}


\subsection{The Dual of the Relaxed Program, and its Min-Cost Flow Interpretation}\label{LP3}

To characterize the structure of optimal solutions to LP2, we use linear programming duality. Consider LP3, LP2's dual program, which has a (flow) variable $f_{S \cup \{i\} \rightarrow S}$ for every set $S$ and $i \notin S$.
\begin{figure}[h!]
\begin{center}
\fbox{\begin{minipage}{13cm}
\[
\min_f \sum_S \sum_{i \notin S} f_{S \cup \{i\} \rightarrow S} d_i
\]
\[
	\begin{array}{lcr}
	\textrm{subject to :} & \\
	\textrm{$\forall S\subseteq N$ : } & -\sum_{i \notin S} f_{S \cup \{i\} \rightarrow S} + \sum_{i \in S} f_{S \rightarrow S \setminus \{i\}} \ge p(S) \left( \sum_{i \in S} x_i - B \right) \\
	\textrm{$\forall S\subseteq N, i \notin S$ : } & f_{S \cup \{i\} \rightarrow S} \ge 0 \\
	\end{array}
\]
\end{minipage}}
\caption{LP3, the dual of LP2}
\end{center}
\end{figure}

\vspace{-10pt}We interpret LP3 as a minimum-cost flow problem on a lattice. Every node on the latice corresponds to a set $S \subseteq N$, and flow may move downwards from $S$ to $S \setminus \{i\}$ for each $i \in S$. The variable $f_{S \rightarrow S \setminus \{i\}}$ represents the amount of flow sent this way, and the cost of sending each unit of flow along this edge is $d_i$.

For nodes $S$ with $ p(S) \left( \sum_{i \in S} x_i - B \right) \geq 0$, we have an external source supplying the node with at least this amount of flow. We call such a node ``positive.''  Nodes with 
$ p(S) \left( \sum_{i \in S} x_i - B \right) < 0$, which we call ``negative,'' can deposit at most   $|  p(S) \left( \sum_{i \in S} x_i - B \right) |$ to an external sink.

Since $d_i > 0$ for all $i$, an optimal solution to LP3 will have net imbalance exactly $ p(S) \left( \sum_{i \in S} x_i - B \right)$ for each positive node $S$.

\section{Characterizing the Linear Programming Solutions} \label{sec:characterization}

For the remainder of the paper, we restrict our attention to the case where $N$ is the only positive node in LP3. We notice that there is a feasible solution if and only if \[
p(N) \left(\sum_{i \in N} x_i - B\right) \le - \sum_{S \subsetneq N}  p(S) \left(\sum_{i \in S} x_i - B\right)\]
which occurs when
$E_S[ \sum_{i \in S} x_i - B ] \le 0.$ Since the components of $S$ are chosen independently, the program is feasible precisely when
$\sum_{i \in N} p_i x_i - B \le 0.$

\subsection{The Canonical Solution to LP3}\label{canonicalsolution}

When there is a single positive node, we can easily construct an optimal solution to LP3 as follows. Define the cost of each node $S$ to be $cost(S) = \sum_{i \in N \setminus S} d_i$, and order the negative nodes $S_1, S_2, \ldots$ in non-decreasing order of cost {(and lexicographically if there are ties)}. We greedily send flow to the negative nodes in order, moving to the next node only when all previous nodes have been saturated. (The flow can be sent along any path to the node, since all such paths have the same cost.) We stop when a net flow of $p(N) \left( \sum_{i \in N} x_i - B \right)$ has been absorbed by the negative nodes we sent flow to. We call this the \emph{canonical solution} to LP3, and notice that the canonical solution is the unique optimal solution to LP3 up to the division of flow between equal cost nodes.

\subsection{From LP3 to LP2 Solutions}\label{solvingrelaxed}

We now show how to use a canonical solution to LP3 to construct a solution to LP2.
In most instances, this solution is unique.
%
%
%
%

\begin{lemma}\label{uniqueness}
Let $S^*$ be the highest-cost negative node which absorbs non-zero flow in the canonical solution $f$  of LP3, and suppose that $S^*$ is not fully saturated by $f$. Then the utility function $u(S) = \max \{ cost(S^*) - cost(S), 0 \}$ is the unique optimal solution to LP2.
\end{lemma}

\begin{proof}
Consider an arbitrary optimal LP2 solution $u$. We will use linear programming complementarity to prove that $u$ is uniquely determined by the canonical solution $f$.

 For any node $S$ that receives nonzero flow in $f$, there is a path $N=S_0, S_1, S_2, \ldots, S_k = S$ from $N$ to $S$ that has positive flow along each edge. By complementarity, the (BIC2) inequalities corresponding to these edges in the primal program are tight {in $u$}. That is, for all $i=1,\ldots, k$, we have $u(S_{i-1}) - u(S_i) = d_x$, where $x$ is the unique element of $S_{i-1} \setminus S_i$. So, for any $S$ which receives nonzero flow in $f$:
$$u(N) - u(S) = \sum_{i \in N \setminus S}d_i = cost(S).$$
For all nodes $S'$ which are not fully saturated in $f$ (i.e. $S^*$ as well as all nodes which receive no flow), $u(S')$ must be 0 in $u$ by complementarity, since the corresponding LP3 constraints are not tight. In particular, since $S^*$ receives flow but is not fully saturated, we have {$u(S^*)=0$ and hence}:
$$u(N)= u(N) - u(S^*) = cost(S^*).$$
Therefore, any node $S$ which receives flow in $f$ must have $u(S) = u(N)-cost(S) = cost(S^*) - cost(S)$.

Furthermore, a node $S$ always receives flow in $f$ if its cost is less than $cost(S^*)$, and receives no flow if its cost is  greater than $cost(S^*)$. {Moreover, those nodes $S$ with $cost(S)=cost(S^*)$ either receive no flow in which case $u(S)=0$, or receive flow in which case $u(S)=cost(S^*) - cost(S)=0$.} Thus, we have shown that for any node $S$, $u(S) = \max \{ cost(S^*) - cost(S), 0 \}$. It is easy to verify that this utility function satisfies all the constraints of LP2.
\end{proof}

If the highest cost node $S^*$ to receive flow in $f$ is fully saturated, then the utility function described above is still an optimal LP2 solution. However, in this case, if the cheapest unfilled node in $f$ has strictly greater cost than $S^*$, then the optimal primal solution is not unique.

\subsection{From LP2 to LP1 Solutions}\label{lp2tolp1}

We now show that, in certain cases, a solution to LP2 allows us to obtain a solution to LP1 where  $\vec{p}$, $\vec{a}$ and $\vec{d}$ are as in Lemma~\ref{newvariables}.

\begin{lemma}\label{lem:supermodulariy}
{Suppose $B > \sum_{i\in N} p_i x_i$ and an optimal solution $u$ to LP2 is monotone and supermodular. Then there is some $q$ such that $(u,q)$ is an optimal solution to LP1 where $\vec{p}$, $\vec{a}$, $\vec{d}$ are as in Lemma~\ref{newvariables}. If $u$ is the unique optimal solution to LP2, then $(u,q)$ is the unique optimal LP1 solution.}
\end{lemma}
\begin{proof}
{We set
\[ q_i(S) =  \left\{
	\begin{array}{ll}
		1,  & \mbox{if } i \in S; \\
		\frac {u(S\cup\{i\}) - u(S)} {d_i}, & \mbox{otherwise.}
	\end{array}
\right. \]
With this choice, as explained in Section~\ref{LP2}, $(u,q)$ is an optimal solution to a relaxation of LP1. So to establish optimality of $(u,q)$ for LP1 it suffices to show that $(u,q)$ satisfies all the constraints of LP1.} 

We first notice that the (IR) constraints are satisfied, since  $u(S) \ge 0$ for all $S$ in LP{2}.

We now show that the (PROB) constraints are satisfied. Indeed, if $i \in S$, then $q_i(S) = 1$. If $i \not\in S$, then $q_i(S) \geq 0$ follows from monotonicity of $u$. The inequality $q_i(S) \leq 1$ follows from constraint (BIC2) of LP2.

Finally, we show that the (BIC) constraints of LP1 are satisfied.
By supermodularity of $u$ we have that for all $S$,  all $i \notin S$ and  all $j\neq i$:
\[ u(S\cup\{i\}\cup\{j\}) - u(S \cup\{j\}) \ge u(S\cup\{i\}) - u(S). \]
Dividing by $d_i$ we obtain $ q_i(S\cup\{j\}) \ge q_i(S)$ for all $i \not\in S$ and $j \neq i$. Since the inequality is trivially satisfied if $i \in S$ (since both sides are 1), or $j = i$ (since $q_i(S \cup \{i\}) = 1$) we conclude that $\vec{q}$ is monotone.

Now pick any distinct subsets $S, T \subseteq N$. We must show that:
\[ u(S) \ge u(T) + (\vec v(S) - \vec v(T)) \cdot \vec q(T).\]
Consider an ordering $i_1, i_2, \ldots, i_k$ of the elements of $T \setminus S$ and an ordering $j_1, j_2, \ldots, j_\ell$ of the elements of $S \setminus T$.

By (BIC2), we know that, for all $r = 1, \ldots, k$:
$$u\left(S \cup \bigcup_{t=1}^{r} \{i_t\}\right) \leq u\left(S \cup \bigcup_{t=1}^{r-1}\{i_t\}\right) + d_{i_r}.$$
Summing over $r$ and cancelling terms, we conclude $u(S \cup T) \leq u(S) + \sum_{r=1}^k d_{i_r}$.

From our definition of $\vec{q}$ it follows that for all $r = 1, \ldots, \ell$:
$$u\left(T \cup \bigcup_{t=1}^{r} \{j_t\}\right) = u\left(T \cup \bigcup_{t=1}^{r-1}\{j_t\}\right) + d_{j_r}q_{j_r}\left(T \cup \bigcup_{t=1}^{r-1}\{j_t\}\right).$$
By monotonicity of $\vec{q}$, it follows that 
$$u\left(T \cup \bigcup_{t=1}^{r} \{j_t\}\right) \geq u\left(T \cup \bigcup_{t=1}^{r-1}\{j_t\}\right) + d_{j_r}q_{j_r}(T).$$
Summing over $r$, we conclude that $u(S \cup T) \geq u(T) + \sum_{r=1}^\ell d_{j_r}q_{j_r}(T). $

Combining this with our earlier upper bound for $u(S \cup T)$, we conclude that
$$u(S) \geq u(T)  + \sum_{r=1}^\ell d_{j_r}q_{j_r}(T) -  \sum_{r=1}^k d_{i_r}.$$
Since $q_{i_r}(T) = 1$ for all $r$, we have
$$u(S) \geq u(T) + \sum_{j \in S \setminus T}d_j q_j(T) - \sum_{i \in T \setminus S}d_i q_i(T)$$
and thus the (BIC) constraint of LP1 is satisfied.

If $u$ is the unique optimal solution to LP2, then {the $(u,q)$ constructed as above is the unique optimal solution to LP1, as it is the unique optimal solution of a relaxation of LP1.}\end{proof}

\subsection{Putting it All Together} \label{sec:end lp theorem}

In summary, we have shown that if the canonical solution of LP3 has a partially saturated node $S^*$, then LP2 has a unique optimal solution, namely $u(S) = \max\{ cost(S^*)-cost(S),0\}$. Since this utility function is monotone and supermodular, it {also defines} a unique optimal solution of the corresponding LP1 instance.

\vspace{-5pt}\begin{corollary}\label{maincor}
Let $S^*$ be the highest-cost negative node which absorbs non-zero flow in the canonical solution of LP3, and suppose that $S^*$ is not fully saturated. Then the original mechanism design problem with $\vec{p}$, $\vec{a}$ and $\vec{d}$ as in Lemma~\ref{newvariables} has a unique optimal solution, and the utility of a player of type $N$ in this solution is $ cost(S^*)$.
\end{corollary}

\section{Proof of Theorem~1: Hardness Of Mechanism Design}\label{hardnessproof} 

We use the results of the previous section to establish the computational hardness of optimal mechanism design. Our reduction is from the lexicographic rank problem, which we show to be {\tt \#P}-hard.

\begin{definition}[{\sc LexRank} problem]
Given a collection  $\mathfrak{C}= \{c_1, \ldots, c_n\}$ of positive integers and a subset $S \subseteq \{1,\ldots, n
\}$, we define the \emph{lexicographic rank} of $S$,  denoted $lexr_\mathfrak{C}(S)$, by
\begin{align*}
lexr_{\mathfrak{C}}(S) \triangleq \vline \bigg\{ S' &:   |S'| = |S| \textrm{ and }   \\
& \left( \sum_{i \in S'} c_i < \sum_{j \in S}c_j \textrm{ or }  \left(\sum_{i \in S'} c_i = \sum_{j \in S}c_j \textrm{ and } S' \leq_{\text{lex}} S \right) \right) \bigg\} \vline
\end{align*}
where $S' \leq_{\text{lex}} S$ is with respect to the lexicographic ordering.\footnote{To be precise, we say that $S_1 \leq_{\text lex} S_2$ iff the largest element in the symmetric difference $S_1 \triangle S_2$ belongs to $S_2$.}
The {\sc LexRank} problem is: Given $\mathfrak{C}$, $S$, and an integer $k$, determine whether or not $lexr_\mathfrak{C}(S) \leq k$.
\end{definition}

\subsection{Hardness of {\sc LexRank}}

We now show that the  {\sc LexRank} problem  is  {\tt \#P}-hard by a reduction from \#-{\sc SubsetSum}.

\begin{definition}[\#-{\sc SubsetSum} problem]
Given a collection $\mathcal{W} = \{w_1,\ldots,w_n\}$ of positive integers and a target integer $T$, compute the number of subsets
$S \subseteq \{1,\ldots, n\}$ such that $\sum_{i \in S} w_i \le T$.
\end{definition}
The \#-\textsc{SubsetSum} problem is known to be {\tt \#P}-hard. Indeed, the reduction from {\sc SAT} to {\sc SubsetSum} as presented in \cite{Sipser2006} is parsimonious.

\begin{lemma}
 {\sc LexRank} is {\tt \#P}-hard.
\end{lemma}

\begin{proof}
Given an oracle for the  {\sc LexRank} problem, it is straightforward to do binary search to compute the lexicographic rank of a set $S$. We will prove hardness of $\textsc{LexRank}$ by reducing the  \#-\textsc{SubsetSum} problem to the computation of lexicographic ranks of a collection of sets.

Let $(\mathcal{W}, T)$ be an instance of \#-\textsc{SubsetSum}, where $\mathcal{W} =  \{w_1,\ldots,w_n\}$ is a collection of positive integers and $T$ is a target integer.
We begin by defining, for $m = 1,\ldots, n$:
$$\textrm{count}_\mathcal{W}(T,m) \triangleq \left| \left\{ S \subseteq \{1,\ldots, n\} : |S| = m \textrm{ and }  \sum_{i \in S}t_i \leq T \right\} \right|. $$
Note that the number of subsets of $\mathcal{W}$ which sum to at most $T$ is simply $\sum_{m=1}^n \textrm{count}_\mathcal{W}(T,m)$. So it suffices to compute $\textrm{count}_\mathcal{W}(T,m)$ for all $m$.

To do this, we define $n$ different collections $\mathfrak{C}_1, \ldots, \mathfrak{C}_n$, where $\mathfrak{C}_{\ell} = \{ c^\ell_1, \ldots, c^\ell_{n+\ell}\}$ is given by:
$$
c^\ell_i = 
\begin{cases}
4nw_i & \mbox{if } 1 \leq i \leq n\\
4nT + 2n & \mbox{if } i = n+1\\
1 & \mbox{if } n+2 \leq i \leq n+\ell.
\end{cases}
$$
We also define a special set $S_\ell \triangleq \{n+1, n+2, \ldots, n+\ell\}$. Notice that $\sum_{i \in S_\ell}c^\ell_i = 4nT + 2n + \ell \costasnote{-1}$. 
Furthermore, for every subset $S \subseteq \{1,\ldots, n\}$, we have
$$\sum_{i \in S} c_i^\ell = 4n \sum_{i \in S} w_i.$$
Hence, for all $\emptyset \neq S \subseteq \{1,\ldots, n\}$:
\begin{enumerate}
\item if $\sum_{i \in S} w_i > T$, then  $\sum_{i \in S} c_i^\ell > \sum_{j \in S_\ell} c_j^\ell $; 

\item if $\sum_{i \in S} w_i \leq T$, then for all $U \subseteq \{n+2, n+3, \ldots, n+\ell\}$ we have  $\sum_{i \in S \cup U} c_i^\ell < \sum_{j \in S_\ell} c_j^\ell $;

\item for all $U \subseteq \{n+2, n+3, \ldots, n+\ell\}$, $\sum_{i \in S \cup U \cup \{n+1\}} c_i^\ell > \sum_{j \in S_\ell} c_j^\ell $.

\end{enumerate}

\smallskip \noindent Given that $|S_\ell| = \ell$ the above imply
\begin{align*}
lexr_{\mathfrak{C}_\ell}\left( S_\ell \right) = 1 + \sum_{m=1}^{\ell}\left( \textrm{count}_\mathcal{W}(T,m)\cdot \binom{\ell - 1}{\ell - m}\right).
\end{align*}

Suppose we have an oracle which can compute the lexicographic rank of a given set. We first use this oracle to determine $lexr_{\mathfrak{C}_1}(S_1)$ and thereby compute $\textrm{count}_\mathcal{W}(T,1)$. Next, we use the oracle to determine $lexr_{\mathfrak{C}_2}(S_2)$ and thereby compute $\textrm{count}_\mathcal{W}(T,2)$, using the previously computed value of $\textrm{count}_\mathcal{W}(T,1)$. Continuing this procedure $n$ times, we can compute $\textrm{count}_\mathcal{W}(T,m)$ for all $m = 1, \ldots, n$. This concludes the proof. \end{proof}

\subsection{Hardness of Mechanism Design: Reduction from {\sc LexRank}}


We will prove hardness of the OMD problem via a reduction from \textsc{LexRank}.
Let  $(\mathfrak{C},S,k)$ be an instance of {\sc LexRank} where  $\mathfrak{C} = \{c_1, \ldots, c_n\}$ is a collection of positive integers, $S \subseteq \{1,\ldots,n\}$, and $k$ is integer. We wish to determine whether $lexr_\mathfrak{C}(S) \leq k$. \costasnote{We assume that $|S| \neq 0, n$ as otherwise the problem is trivial to solve.}

We denote by $[n]$ the set $\{1,2,\ldots, n\}$ and $[n+1] = [n] \cup \{n+1\}$. We construct an OMD instance indirectly, by defining an instance of LP3 with the following parameters:
\begin{itemize}
	\item $d_i = 2^{n+1}\left(c_i + \sum_{j=1}^n c_j \right) + 2^{i}$, for $i = 1, \ldots, n$;
	\item $d_{n+1} = 1$;
	\item $x_i = 2$, for all $i$;
	\item $B = 2n+1$;
	\item $p_i = p$ for all $i$, where we leave $p \in [0.5, 1 - \frac{1}{2n+2})$ a parameter.
\end{itemize}
We note that $B > \sum_i x_i p_i$, and thus Lemma~\ref{newvariables} implies that, for all $p$, an instance of LP3 as above arises from some OMD instance $\{a_i,d_i,p_i\}_{i=1}^{n+1}$, in the notation of Section~\ref{sec:prelim}.

Denote by $S^c$ the set $[n] \setminus S$.\footnote{Note that $\{n+1\}$ is in neither $S$ nor $S^c$.} Suppose that, for some value $p$, there is a partially filled node $T^*$ in the canonical LP3 solution such that $T^* \subseteq [n]$ and $|T^*| = n - |S|$. Using Lemma~\ref{uniqueness} we have
\begin{align*}
q_{n+1}(S^c) &= \frac{u^*(S^c \cup \{n+1\}) - u^*(S^c)}{1} \\
&= \max\{cost(T^*) - cost(S^c \cup \{n+1\}),0 \} -  \max\{cost(T^*) - cost(S^c),0 \}\\
&=  \max\{cost(T^*) - cost(S^c) + 1,0 \} -  \max\{cost(T^*) - cost(S^c),0 \}
\end{align*}
Therefore, since the cost of each set is an integer,
$$
q_{n+1}(S^c) = 
\begin{cases}
0 & \mbox{if } cost(S^c) > cost(T^*)\\
1 & \mbox{if } cost(S^c) \leq cost(T^*).
\end{cases}
$$
Since $n+1$ is in neither $S^c$ nor $T^*$,
$$
q_{n+1}(S^c) = 
\begin{cases}
0 & \mbox{if } \sum_{i \in S}d_i  > \sum_{j \in [n] \setminus T^*} d_j \\
1 & \mbox{if } \sum_{i \in S}d_i  \leq \sum_{j \in [n] \setminus T^*} d_j .
\end{cases}
$$
By our construction of the $d_i$'s we can see that since $|T^*| = n - |S|$,
$$
q_{n+1}(S^c) = 
\begin{cases}
0 & \mbox{if } \sum_{i \in S}c_i  > \sum_{j \in [n] \setminus T^*} c_j \\
1 & \mbox{if } \sum_{i \in S}c_i  < \sum_{j \in [n] \setminus T^*} c_j \\
1 & \mbox{if } \sum_{i \in S}c_i = \sum_{j \in [n] \setminus T^*} c_j \textrm{ and } S \leq_{\text{lex}} ([n] \setminus T^*)\\
0 & \mbox{if } \sum_{i \in S}c_i = \sum_{j \in [n] \setminus T^*} c_j \textrm{ and } S >_{\text{lex}} ([n] \setminus T^*)
\end{cases}
$$
Therefore, $q_{n+1}(S^c) = 1$ if $lexr_\mathfrak{C}(S) \leq lexr_{\mathfrak{C}}([n] \setminus T^*)$ and 0 otherwise. 

So our next goal  is to set the parameter $p$ such that \costasnote{there is a partially filled node $T^*$ in the canonical LP3 solution such that $T^* \subseteq [n]$, $|T^*| = n - |S|$, and $lexr_\mathfrak{C}([n] \setminus T^*) = k$}. For such $p$, distinguishing between $q_{n+1}(S^c) = 0$ and $q_{n+1}(S^c) = 1$ would allow us to solve the \textsc{LexRank} instance. The next lemma shows that a $p$ as required can be found in polynomial time.

\begin{lemma}\label{lem:choice of parameter}
In polynomial time, we can identify a $\tilde{p} \in [0.5, 1 - \frac{1}{2n+2})$ with $O(n \log n)$ bits of precision such that the partially filled node in the canonical LP3 solution with parameter $p=\tilde{p}$ is a set $T^* \subseteq [n]$ of size $n-|S|$ and  $lexr_{\mathfrak{C}}([n] \setminus T^*) = k$.
\end{lemma}

\begin{proof}
In our construction, the lowest cost negative node is $[n]$. Furthermore, the cost of every negative node is unique, and for any $T \subsetneq [n]$ there is no node with cost between that of $T \cup \{n+1\}$  and $T$. Also, if $T$ and $T'$ are proper subsets of $[n]$ and if $|T| > |T'|$, then $cost(T) < cost(T')$.

For each $i$ between $1$ and $n-1$, let $T^i_1, T^i_2, \ldots$ be the ordering of the size-$i$ subsets of $[n]$ in increasing order of cost. In the canonical LP3 solution,
 $[n]$ fills first, and $T_j^i$ fills before $T_{j'}^{i'}$ if it has larger size ($i > i'$) or the same size but smaller cost ($i = i'$ and $j < j'$). Furthermore, each node $T^i_j \cup \{n+1\}$ fills immediately before the node $T^i_j$. 

Our goal is to choose $p$ so that $T^{n-|S|}_k$ is partially filled. 
Indeed, the sets $[n]\setminus T_1^{n-|S|}$ through $[n]\setminus T_k^{n-|S|}$ are precisely the sets counted in the computation of $lexr_\mathfrak{C}([n]\setminus T_k^{n-|S|})$. Notice that lexicographic tie-breaking of $lexr$ is enforced by construction of adding an additional $2^i$ to $d_i$.



The only positive node, $[n+1]$, emits a net flow of $p^{n+1}$, and the node $[n]$  absorbs at most $p^n(1-p)$ flow. For each size $i$ between $n-1$ and $n-|S|+1$, there are $\binom{n}{i}$ sets $T \subseteq [n]$ of size $i$, each of which can absorb
$$|p(T)(2|T| - B)| = (2(n-i) +1)p^i(1-p)^{n+1-i}$$
flow. Furthermore, each set $T \cup \{n+1\}$ can absorb
$$|p(T \cup \{n+1\})(2|T| + 2 - B)| = (2(n-i) - 1)p^{i+1}(1-p)^{n-i}.$$
Thus, in total, $T$ and $T \cup \{n+1\}$ can absorb
$$p^i(1-p)^{n-i}\left((1-p)(2n-2i+1) + p(2n-2i-1) \right)$$
$$= p^i(1-p)^{n-i}(2(n-i-p)+1).$$
Finally, we notice that $T$ is responsible for at least a $1/(2n+2)$ fraction of the quantity above, since
$$\frac{(2(n-i)+1)p^i(1-p)^{n+1-i}}{ (2(n-i-p) + 1)p^{i}(1-p)^{n-i}} \geq 1-p > \frac{1}{2n+2},$$
using that $p<1-\frac{1}{2n+2}$.

If all nodes strictly preceding (i.e. with smaller cost than) $T^{n-|S|}_1 \cup \{n+1\}$ have been saturated, the amount of flow still unabsorbed is
$$p^{n+1} - p^n(1-p) -\!\!\! \sum_{i = n-|S|+1}^{n- 1}\binom{n}{i}p^i(1-p)^{n-i}(2(n - i -p) +1)$$
$$=  \sum_{i = n-|S|+1}^{n}\binom{n}{i}p^i(1-p)^{n-i}(2(i +p - n) -1).$$
Therefore, a sufficient condition for $T_k^{n-|S|}$ to be partially filled in the canonical solution is
$$f(p) \triangleq \frac{\sum_{i = n-|S|+1}^{n}\binom{n}{i}p^i(1-p)^{n-i}(2(i +p - n) -1)}{p^{n-|S|}(1-p)^{|S|}(2|S| -2p +1)} \in \left(k- \frac{1}{2n+2},k\right).$$

We claim that there is such a $p^* \in [0.5,1-\frac{1}{2n+2})$ such that $f(p^*) = k - \frac{1}{4n+4}$. Indeed, for $p < 0.5$, only $[n]$ will ever receive flow, so in this case, $f(p) < 0$. Furthermore, we can lower-bound $f(p)$ by the following ratio (where we add a negative quantity to the numerator)
$$\frac{\sum_{i = n-|S|+1}^{n}\binom{n}{i}p^i(1-p)^{n-i}(2(i +p - n) -1) + \sum_{i=0}^{n-|S|-1} \binom{n}{i}p^{n-i}(1-p)^{i}(2(i+p-n)-1)     }{p^{n-|S|}(1-p)^{|S|}(2|S| -2p +1)}$$
$$=\frac{\sum_{i = 0}^{n}\binom{n}{i}p^i(1-p)^{n-i}(2(i +p - n) -1) - \binom{n}{|S|}p^{n-|S|}(1-p)^{|S|}(-2|S|+2p-1)     }{p^{n-|S|}(1-p)^{|S|}(2|S| -2p +1)}$$
and thus
\begin{align*}
f(p) &\geq \binom{n}{|S|} - \frac{\sum_{i = 0}^{n}\binom{n}{i}p^i(1-p)^{n-i}( 2n-2i -2p +1)    }{p^{n-|S|}(1-p)^{|S|}(2|S| -2p +1)}\\
&= \binom{n}{|S|} - \frac{ 2n - 2p + 1 -2 \sum_{i = 0}^{n}i \binom{n}{i} p^i(1-p)^{n-i}   }{p^{n-|S|}(1-p)^{|S|}(2|S| -2p +1)}\\
&= \binom{n}{|S|} - \frac{ 2n - 2p + 1 - 2pn    }{p^{n-|S|}(1-p)^{|S|}(2|S| -2p +1)}.
\end{align*}
Hence, for $p = 1 - \frac{1}{2n+2}$, we get $f(p) \ge \binom{n}{|S|} \ge k$. Using this, the continuity of $f$, and that $f(p)<0$ for $p<0.5$, we conclude that there is a $p^* \in [0.5, 1-\frac{1}{2n+2})$ such that $f(p^*) = k  - \frac{1}{4n+4}$.

We now consider $\tilde{p} = p^* \pm \epsilon \in [0.5, 1-\frac{1}{2n+2})$. We claim that $f(\tilde{p}) \in (k-\frac{1}{2n+2},k)$ as long as $\epsilon = O\left(\frac{(4n)^{-4n}}{4n+4}\right)$. To show this, we bound the absolute value of the derivative $\frac{df}{dp}$ at all points in $[0.5, 1-\frac{1}{2n+2})$. The numerator of $\frac{df}{dp}$ is a polynomial in $p$ of degree $2n+1$, where the coefficient of each term is, in absolute value, $O(2^{3n} \cdot {\rm poly}(n)) \le O(2^{4n})$---using the crude bound ${n \choose i} \le 2^n$. Furthermore, the denominator $(p^{n-|S|}(1-p)^{|S|}(2|S| -2p +1))^2$ of $\frac{df}{dp}$ is greater than $\left(\frac{1}{2n+2} \right)^{2n}$, since $p \in [0.5, 1-\frac{1}{2n+2})$. Therefore, we can bound the magnitude of the derivative by $O((2n+2)^{2n} (2n+1) 2^{4n}) = O((4n)^{4n})$. Since this bound holds for all points in $[0.5, 1-\frac{1}{2n+2})$, we conclude that 
$$f( \tilde{p} ) \in (f( p^*) - \epsilon O((4n)^{4n}), f( p^*) + \epsilon O((4n)^{4n}))$$
and thus, for $\epsilon = O\left(\frac{(4n)^{-4n}}{4n+4}\right)$ we have that $f(\tilde{p} ) \in (k-\frac{1}{2n+2},k)$.

Thus, we can find the desired $\tilde{p}$ in polynomial time via binary search on $f(p)$, requiring $O(n \log n)$ bits of precision in $p$.\end{proof}

\costasnote{We conclude our reduction as follows: First, we compute $\tilde{p}$ as in Lemma~\ref{lem:choice of parameter}. Next we solve the OMD instance resulting from this choice of parameter. The solution (see Section~\ref{sec:interpretation}) induces expected polynomial-time algorithms $\cal C$ and $\cal S$. We use them to sample (in expected polynomial-time) from the allocation and price rule of the optimal auction for type $S^c$. We have proven that this type receives item $n+1$ with probability $1$ if $(\mathfrak{C},S,k) \in \textsc{LexRank}$ and with probability $0$ otherwise. So a single sample from the allocation rule of the optimal auction suffices to tell which one is the case, since if $q_{n+1}(S^c)=1$ then every sample from the allocation rule should allocate the item and if $q_{n+1}(S^c)=0$ then no sample should allocate the item. So if we can solve OMD in expected polynomial-time, then we can solve \textsc{LexRank} in expected polynomial-time by a. finding $\tilde{p}$ using Lemma~\ref{lem:choice of parameter}; b. solving the resulting OMD instance; c. drawing a single sample from the allocation rule of the auction for type $S^c$ and outputting ``yes'' if and only if item $n+1$ is allocated by the drawn sample.~$\qed$}

\section{Beyond Additive or Computationally Bounded Bidders}\label{beyondadditive}

\subsection{Budget-Additive Bidders} \label{sec:budget additive bidders}

We have already remarked in Section~\ref{sec:related} that when the bidder valuations have combinatorial structure, it becomes much easier to embed computationally hard problems into the optimal mechanism design problem.  In this spirit, Dobzinski et al.~\cite{DobzinskiFK11} show that optimal mechanism design for OXS bidders is ${\tt NP}$-hard. While we won't define OXS valuations, we illustrate the richness in their structure by recalling that in~\cite{DobzinskiFK11} the items are taken to be edges of a graph $G=(V,E)$, and there is a single bidder whose valuation is drawn from a distribution that includes in its support the following valuation:
$$f(E')=\max\left\{|A|~~\vline~~\begin{minipage}{6.7cm}$A \subseteq E'$, {every connected component of \\$G'=(V,A)$ is either acyclic or unicyclic} \end{minipage}\right\}, \forall E' \subseteq E.$$

While the valuations used in~\cite{DobzinskiFK11} are quite rich, we notice here that the techniques of~\cite{DobzinskiFK11} can be used to establish hardness of optimal mechanism design for a bidder with very mild combinatorial structure, namely budget-additivity. Indeed, consider the problem of selling a set $N = \{1,\ldots, n\}$ of items to a budget-additive bidder whose value $x_i$ for each item is some  deterministically known integer, but whose budget is only probabilistically known. In particular, suppose that, with probability $(1-\epsilon)$, the bidder is additive, in which case her value $v_{a}(S)$ for each subset $S$ is $\sum_{i \in S}x_i$. With probability $\epsilon$, however, she has a positive integer budget $B$:  she values each subset $S$ at $v_b(S) = \min\{ \sum_{i \in S}x_i, B\}$. That is, when she has a budget, she receives at most $B$ utility from any subset. We claim that the optimal mechanism  satisfies:

\begin{claim}\label{budgetlemma}
Suppose 
  $\epsilon < \frac{1}{1+\sum x_i}$. Then every optimal individually rational and Bayesian incentive compatible direct mechanism for the budget-additive bidder described above has the following form:
\begin{itemize}
	\item If the bidder is unbudgeted, she receives all items and is charged $\sum_{i \in N} x_i$.
	\item If the bidder is budgeted, she receives a probability distribution over the subsets $T$ of items such that $\sum_{i\in T} x_i$ is as large a value as possible without exceeding $B$. She is charged that value.
\end{itemize}
\end{claim}
This claim follows from a lemma of \cite{DobzinskiFK11}, showing that {if the bidder has no budget she must receive her value maximizing bundle, while if she is budgeted she must receive some bundle $T$ maximizing $v_{b}(T) - (1-\epsilon)v_a(T)$}. Determining the optimal direct mechanism is clearly hard in this context, as it is {\tt NP}-hard to compute a  subset $T$ with the largest value that does not exceed $B$. {Using the same arguments as in Section~\ref{sec:prelim} we can extend this lower bound to all mechanisms, by noticing that \costasnote{any sample from the allocation} to the budgeted bidder answers whether there exists some $T$ such that $\sum_{i \in T} x_i =B$.}

\begin{theorem} \label{thm:main2}
There is a polynomial-time Karp reduction from the subset-sum problem to the optimal mechanism design problem of selling multiple items to a single  budget-additive quasi-linear bidder whose values for the items are known rational numbers, and whose budget is equal to some finite rational value or $+\infty$ with rational probabilities.
\end{theorem}

\subsection{Further Discussion: Powerful Bidders} \label{sec:powerfull bidders}

\subsubsection{{\large \tt NP} Power}
We observe that if the budget-additive bidder described in the previous section has the ability to solve {\tt NP}-hard problems, then there is a simple indirect mechanism which implements the above allocation rule: The seller offers each set $S$ at price $\sum_{i \in S}x_i$, and leaves the computation of the set $T$ to the bidder. (As remarked in Section~\ref{sec:prelim} we can use small rewards to guarantee that the same allocation rule is implemented.) Such a mechanism is intuitively unsatisfactory, however, and violates our requirement (see Section~\ref{sec:interpretation}) that a bidder be able to compute her  strategy efficiently. While an optimal indirect mechanism is easy for the seller to construct and to implement, it is intractable for the bidder to determine her optimal strategy in such a mechanism. Conversely, implementing the direct mechanism requires the seller to solve a subset sum instance. Thus, shifting from a direct to an indirect mechanism allows for a shift of computational burden from the seller to the bidder.

\subsubsection{{\large \tt PSPACE} and Beyond}
Indeed, we can generalize this observation to show that optimal mechanism design for a polynomial-time seller becomes much easier if {the bidder is quasi-linear, computationally unbounded, and has} a known finite maximum possible valuation for an allocation. The intuition is that, if some canonical optimal direct mechanism is computable and implementable in polynomial space, then we can construct an extensive form indirect mechanism whereby the bidder first declares her type and then convinces the seller, via an interactive proof with completeness 1 and low soundness, what the allocation and price would be {for her type} in the canonical optimal direct revelation BIC and IR mechanism.\footnote{We have fixed a canonical direct mechanism to avoid complications arising when the optimal mechanism is not unique.} In the event of a failure, the bidder is charged a large fine, significantly greater than her maximum possible valuation of any subset. The proof that this protocol achieves, in a subgame perfect equilbrium, expected revenue equal to the optimal direct BIC and IR revenue  follows from ${\tt IP} = {\tt PSPACE}$~\cite{Shamir:1992:IP:146585.146609}. In particular, we note that the canonical solution of the mechanism design instances in the hardness proof of Section~\ref{hardnessproof} can indeed be found in {\tt PSPACE}, and thus we can construct an easily-implementable extensive form game which achieves optimal revenue in subgame perfect equilibrium. Reaching this equilibrium, however, requires the bidder to have {\tt PSPACE} power.

Indeed, if there are two or more computationally-unbounded bidders, then we can replace the assumption that the optimal direct mechanism is computable in {\tt PSPACE} with the even weaker assumption that it is computable in nondeterministic exponential time, using the result that ${\tt MIP} = {\tt NEXP}$~\cite{DBLP:journals/cc/BabaiFL91} to achieve the optimal revenue in a perfect Bayesian equilibrium. The intuition is that our mechanism now asks all players to (simultaneously) declare their type. Two players are selected and required to declare the allocation and price which the canonical optimal direct mechanism would implement given this type profile, and they then perform a multiparty interactive proof to convince the seller of the correctness of the named allocation and price.

Since the optimal allocation might be randomized and the {prices} might require exponentially many bits to specify, we must modify the above outline. The rough ideas are the following. First, a price of exponential bit complexity but magnitude bounded by a given value can be obtained as the expectation of a distribution that can be sampled with polynomially many random bits in expectation. Now, after the bidders declare their types, the seller reveals a short random seed $r$, and requires the bidders to draw the prices and allocation from the appropriate distribution, using $r$ as a seed. The bidders must prove the actual deterministic allocation and prices that the optimal direct mechanism would have obtained  when given seed $r$ to sample exactly from the optimal distribution. If $r$ does not have sufficient bits to determine the sample, then the bidders prove interactively that additional bits are needed, and further random bits are appended to $r$. If the bidders ever fail in their proof, the mechanism terminates and they are charged a large fine. Regardless of whether or not the players ever attempt to prove an incorrect statement, the protocol  terminates in expected polynomial time. 

We omit a formal proof of the above observations. {We also note that we do not present the mechanisms of this section as practical. But we want to point out that the complexity of the optimal mechanism design problem may become trivial if no computational assumptions are placed on the bidders.}

\vspace{-5pt}\section{Conclusion}
Prior to this work, the complexity of computing revenue-optimal mechanisms was poorly understood despite considerable effort in both economics and algorithmic game theory. Exact solutions were known only in restricted classes of instances or for constrained input models, yet complexity lower bounds were known only for instances which injected rich combinatorial structure on the valuations or the allowable allocations. In this paper, we provide compelling lower bounds for one of the simplest settings:  a single additive bidder with independent values for the items. Our lower bound applies to the most general computational model, placing no restriction on the type or encoding of the mechanism, except that its outcomes be efficiently samplable. The proof technique is in itself interesting, developing relaxations to a folklore LP for revenue-optimization to narrow into a family of instances that are rich enough to encode {\tt \#}P-hard problems, yet are amenable to analysis.
 
We note that our results do not preclude the existence of a fully polynomial-time approximation scheme (FPTAS) for the mechanism design problem. Indeed, this remains an open question even in the case of a single additive bidder. \vspace{-8pt}

\bibliographystyle{abbrv}
\bibliography{costasbibold}  
\end{document}